\documentclass[10pt,a4paper]{llncs}

\usepackage{palatino}
\usepackage[T1]{fontenc}
\usepackage[utf8]{inputenc}
\usepackage{units}
\usepackage{amsmath}
\usepackage{graphicx}
\usepackage{amssymb}
\usepackage{subfigure}
\usepackage{multirow}
\usepackage{colortbl} 
\usepackage{xcolor}
\usepackage{floatflt}
\usepackage{fancyhdr}
\usepackage{setspace}
\usepackage{graphicx} 
\usepackage{enumitem}

\usepackage{amsfonts,hyperref,color}
\usepackage{times}
\usepackage{tikz}
\usetikzlibrary{arrows,automata}
\usetikzlibrary{matrix}
\usepackage{booktabs}

\newcommand{\FF}{\mathbb{F}}

\newcommand{\dd}{\mathrm{d}}

\newcommand{\len}{\mathrm{len}}

\title{On the Griesmer bound for nonlinear codes}

\author{Emanuele Bellini\inst{1} \and Eleonora Guerrini\inst{2} \and Alessio Meneghetti\inst{3}  \and Massimiliano Sala\inst{3}}

\institute{
Telsy S.p.A., Italy\\
\email{eemanuele.bellini@gmail.com}
\and
LIRMM, Universit\'e
de Montpellier 2, France\\
\email{guerrini@lirmm.fr}
\and
University of Trento, Italy\\
\email{almenegh@gmail.com}
\and
University of Trento, Italy\\
\email{maxsalacodes@gmail.com}
} 

\begin{document}
\maketitle
\begin{abstract}
Most bounds on the size of codes hold for any code, whether linear or nonlinear. 
Notably, the Griesmer bound, holds only in the linear case.
In this paper we characterize a family of systematic nonlinear codes for which the Griesmer bound holds. 
Moreover, we show that the Griesmer bound does not necessarily hold for a systematic code
 by showing explicit counterexamples.
On the other hand, we are also able to provide (weaker) versions of the Griesmer bound holding for all systematic codes.
\end{abstract}


\section{Introduction}
We consider codes over a finite field $\FF_q$ of length $n$, with $M$ codewords, and distance $d$. 
A code $C$ with such parameters is denoted as an $(n,M,d)_q$-code.
\begin{definition}\label{defn: systematic}
An $(n,q^k,d)_q$-systematic code $C$ is the image of a map $F:\left(\FF_q\right)^k\to \left(\FF_q\right)^n$, $n\ge k$, s.t.
a vector $x = (x_1,\ldots,x_k) \in (\FF_q)^k$ is mapped to a vector $$(x_1,\ldots,x_k,f_{k+1}(x),\ldots,f_n(x)) \in (\FF_q)^n,$$
where 
$f_i,i=k+1,\ldots,n$ are maps from $(\FF_q)^k$ to $\FF_q$. 
We refer to $k$ as the dimension of $C$. 
The coordinates from 1 to $k$ are called systematic, while those from $k+1$ to $n$ are called non-systematic.
\end{definition}
If the maps $f_i$ are all linear, then the systematic code $C$ is a subspace of dimension $k$ of $(\FF_q)^n$ and we say it is a  $[n,k,d]_q$-linear code. 
A nonlinear code is a code which is not necessarily linear or systematic.\\
We denote with $\len(C),\dim(C),\dd(C)$, respectively, the length, the dimension (when defined) and the minimum distance of a code $C$.\\
A central problem of coding theory is to determine the minimum value of $n$, for which an $(n,M,d)_q$-code or an $[n,k,d]_q$-linear code exists. 
We denote by $N_q(M,d)$ the minimum length of a nonlinear code over $\FF_q$, with $M$ codewords and distance $d$. 
We denote by $S_q(k,d)$ the same value in the case of a systematic code of dimension $k$, while we use $L_q(k,d)$ in the case of a linear code of dimension $k$. Observe that 
$$
N_q\left(q^k,d\right)\leq S_q(k,d)\leq L_q(k,d).
$$
A well-known lower bound for $L_q(k,d)$ is
\begin{theorem}[Griesmer bound]\label{thm:griesmer_linear}
All $[n,k,d]_q$ linear codes satisfy the following bound:
\begin{equation}\label{eq:griesmer_linear}
n \ge L_q(k,d)\ge  g_q(k,d) := \sum_{i=0}^{k-1}\left\lceil\frac{d}{q^i}\right\rceil
\end{equation}
\end{theorem}
The Griesmer bound, which can be seen as an extension of the Singleton bound ($n \ge d+k-1$) \cite{CGC-cd-book-huffmanPless03} (Section 2.4) in the linear case, 
has been introduced by Griesmer \cite{CGC-cd-art-griesm60} in the case of binary linear codes and then 
generalized by Solomon and Stiffler \cite{CGC-cod-art-solstiffl65} in the case of $q$-ary linear codes. \\
It is known that the Griesmer bound is not sharp \cite{CGC-cod-art-maruta1996non}, \cite{CGC-cod-art-van1980uniqueness}, \cite{CGC-cod-art-Maruta1997}. \\
Important examples of linear codes meeting the Griesmer bound are the simplex code \cite{CGC-cd-book-huffmanPless03} (Section 1.8) and the $[5,6,11]_3$ Golay code \cite{CGC-cd-book-huffmanPless03} (Section 1.9), \cite{CGC-gola49}. \\
Many authors such as
\cite{CGC-cod-art-helleseth1981characterization}, 
\cite{CGC-cod-art-hamada1993characterization}, 
\cite{CGC-cod-art-tamari1984linear}, 
\cite{CGC-cod-art-Maruta1997}, and 
\cite{CGC-cod-art-Klein04}, 
have characterized classes of linear codes meeting the Griesmer bound.
In particular, finite projective geometries play an important role in the study of these codes. 
For example in 
\cite{CGC-cod-art-hellesth1992projective}, 
\cite{CGC-cod-art-Hamada93} and 
\cite{CGC-cod-art-Tamari93} 
minihypers and maxhypers are used to characterize linear codes meeting the Griesmer bound. 
Research has been done also to characterize the codewords of linear codes attaining the Griesmer bound \cite{CGC-cod-art-ward1998divisibility}.\\
Many known bounds on the size of nonlinear codes, for example 
the Johnson bound (\cite{CGC-cd-art-john62},\cite{CGC-cd-art-john71},\cite{CGC-cd-book-huffmanPless03}), 
the Elias-Bassalygo bound (\cite{CGC-cd-art-bass65},\cite{CGC-cd-book-huffmanPless03}), 
the Levenshtein bound (\cite{CGC-cd-art-lev98}), 
the Hamming (Sphere Packing) bound, 
the Singleton bound (\cite{CGC-cd-book-pless98}),  
the Plotkin bound (\cite{CGC-cd-art-plotk60}, \cite{CGC-cd-book-huffmanPless03}),
the Zinoviev-Litsyn-Laihonen (\cite{CGC-cod-art-zinlyts1984shortening}, \cite{CGC-cod-art-litlai98}),
the Bellini-Guerrini-Sala (\cite{CGC-cd-art-BellGuerrSal2014}), and
the Linear Programming bound (\cite{CGC-cd-art-Dels73}),
are true for both linear and nonlinear codes. \\
\indent
The proof of the Griesmer bound heavily relies on the linearity of the code and it cannot be applied to all nonlinear codes. 
%
%
In Section \ref{sec:systematic_nonlinear_codes} we prove that,
once $q$ and $d$ have been chosen, if all nonlinear $(n,q^{k},d)_q$-systematic codes with $k<1+\log_qd$ respect the Griesmer bound, then the Griesmer bound holds for all systematic codes with the same $q$ and $d$. In particular for any $q$ and $d$ only a finite set of $(k,n)$-pairs has to be analysed in order to prove the bound for all $k$ and $n$.\\
Using this result, in Section \ref{sec:systematic_griesmer} we characterize several families of systematic codes for which the Griesmer bound holds. 
In Section \ref{sec:weak_griesmer} we provide (weak) versions of the Griesmer bound, holding for any systematic code. Finally, in Section \ref{sec:counterexample}, we show explicit counterexamples of nonlinear codes and systematic codes for which the Griesmer bound does not hold.
%
%
%


\section{A sufficient condition to prove the Griesmer bound for systematic codes}\label{sec:systematic_nonlinear_codes}
The following proposition is well-known, we however provide a sketch of the proof for the particular case in which we will make use of it.
\begin{proposition}
Let $C$ be an $(n,q^k,d)$-systematic code, and $C'$ be the code obtained by shortening $C$ in a systematic coordinate. Then $C'$ is an $(n-1,q^{k-1},d')$-systematic code with $d'\ge d$.
\end{proposition}
\begin{proof}
To obtain $C'$, consider the code $C'' = \left\{F(x) \mid x = (0,x_2,\ldots,x_k) \in \left(\FF_q\right)^k \right\}$, i.e. the subcode of $C$ which is the image of the set of messages whose first coordinate is equal to $0$. 
Then $C''$ is such that $\dim(C'') = k-1$ and $\dd(C'') \ge d$. 
Since, by construction, all codewords have the first coordinate equal to zero, we obtain the code $C'$ by puncturing $C''$ on the first coordinate, so that $\len(C')=n-1$ and $d' = \dd(C') = \dd(C'') \ge d$.
\end{proof}
The following lemma is well-known, but we provide a proof because it anticipates our later argument.
\begin{lemma}\label{lem: reducing d}
If $n > k$, then given an $(n,q^k,d)$-systematic code $C$, there exists an $(n,q^k,\bar{d})$-systematic code $\bar{C}$ for any $1\leq\bar{d}\leq d$. 
\end{lemma}
\begin{proof}
If $n > k$, we can consider the code $C^1$ obtained by puncturing $C$ in a non-systematic coordinate.  
$C^1$ is an $(n-1,q^k,d^{(1)})$-systematic code. Of course, either $d^{(1)}=d$ or $d^{(1)}=d-1$.\\
By puncturing at most $n-k$ coordinates, we will find a code whose distance is $1$. Then there must exists an $i\leq n-k$ such that the code $C^i$, obtained by punturing $C$ in the last $i$ coordinates, has distance equal to $\bar{d}$.
%
\end{proof}
%
\begin{theorem}\label{thm: minimum k q}
For fixed $q$ and $d$, if
\begin{equation} \label{eq:griesmer_inequality}
 S_q(k,d) \ge g_q(k,d)
\end{equation}
for all $k$ such that $1 \le k < 1+\log_q d$, 
then \eqref{eq:griesmer_inequality} holds for any positive $k$,
i.e. the Griesmer bound is true for all the systematic codes over $\FF_q$ with minimum distance $d$.
\end{theorem}
\begin{proof}
It is sufficient to show that 
 if an $(n,q^k,d)_q$-systematic code not satisfying the Griesmer bound exists, 
 then an $(n',q^{k'},d)_q$-systematic code not satisfying the Griesmer bound exists with $k'<1+\log_q d$, and $n'>k'$. \\
 For each fixed $d,q$ suppose there exists an $(n,q^k,d)_q$-systematic code not satisfying the Griesmer bound, i.e., there exists $k$ such that $S_q(k,d) < g_q(k,d)$. \\
 Let us call 
 $\Lambda_{q,d} = \{k \ge 1 \mid S_q(k,d) < g_q(k,d) \}$.\\
 If $\Lambda_{q,d}$ is empty than the Griesmer bound is true for such parameters $q,d$.\\
 Otherwise there exists a minimum $k'\in\Lambda_{q,d}$ such that $S_q(k',d) < g_q(k',d)$.\\
 In this case we can consider an $(n,q^{k'},d)_q$ systematic code $C$ not verifying the Griesmer bound,
 namely $n = S_q(k',d)$.
 Due to Definition \ref{defn: systematic}, $C$ can be seen as the image of a map
 $$F(x)=(x_1,\ldots,x_{k'},f_{k'+1}(x),\ldots,f_n(x)),$$ where $x=(x_1,\ldots,x_{k'})$.
 We define a code $C'$ as the image of $$ F(x')=(x_2,\ldots,x_{k'},f_{k'+1}(0,x_2,\ldots,x_{k'}),\ldots,f_n(0,x_2,\ldots,x_{k'}))$$ where $x'=(x_2,\ldots,x_{k'})$.
 Clearly, $C'$ is an $(n-1,q^{k'-1},d')$ systematic code and $d'\ge d$. 
 Applying Lemma \ref{lem: reducing d} to $C'$, we can obtain an $(n-1,q^{k'-1},d)_q$ systematic code $\bar{C}$. 
 Since $k'$ was the minimum among all the values in $\Lambda_{q,d}$, then the Griesmer bound holds for $\bar{C}$, and so 
 \begin{equation}\label{eq: griesmer n-1}
 n-1 \ge g_q(k'-1,d) = \sum_{i=0}^{k'-2}\left\lceil\frac{d}{q^i}\right\rceil.
 \end{equation}
 We observe that, if $q^{k'-1}\ge d$, then $\left\lceil\frac{d}{q^{k'-1}}\right\rceil=1$, so we can rewrite \eqref{eq: griesmer n-1} as
 \begin{align*}
  n \ge \sum_{i=0}^{k'-2}\left\lceil\frac{d}{q^i}\right\rceil +1\ge \sum_{i=0}^{k'-2}\left\lceil\frac{d}{q^i}\right\rceil + \left\lceil\frac{d}{q^{k'-1}}\right\rceil 
      = \sum_{i=0}^{k'-1}\left\lceil\frac{d}{q^i}\right\rceil
      = g_q(k',d)
 \end{align*}
 Since we supposed $n < g_q(k',d)$, we have reached a contradiction with the assumption $q^{k'-1}\ge d$. Hence for such $d$, the minimum $k$ in $\Lambda_{q,d}$ has to satisfy $q^{k-1}<d$, which is equivalent to our claimed expression $k<1+\log_q d$.
\end{proof}
\section{Set of parameters for which the Griesmer bound holds in the nonlinear case}
\label{sec:systematic_griesmer}
In this section we characterize several sets of parameters $(q,d)$ for which the Griesmer bound holds for systematic codes.
\subsection{The case $d \le 2q$}
We use Proposition \ref{thm: minimum k q} to prove that all $q$-ary systematic codes with distance up to $2q$ satisfy the Griesmer bound.
\begin{theorem}\label{thm: griesmer 2q}
If $d\le 2q$ then $S_q(k,d)\ge g_q(k,d)$.
\end{theorem}
\begin{proof}
First, consider the case $d \leq q$. 
By Theorem \ref{thm: minimum k q} it is sufficient to show that, fixing $q,d$, for any $n$ an $(n,q^k,d)_q$-systematic code with $1 \le k<1+\log_q d$ and $n < g_q(k,d)$ does not exists. 
If $1 \le k<1+\log_q d$ then $\log_q d \le \log_q q = 1$, and so $k$ may only be 1. 
Since $g_q(1,d) = d$ and $n \ge d$, we clearly have that $n \ge g_q(1,d)$.\\
%
%
Now consider the case $q<d\leq 2q$. 
We use again Theorem \ref{thm: minimum k q},
i.e. we show that, fixing $q,d$, then for any $n$ an $(n,q^k,d)_q$-systematic code with $1 \le k<1+\log_q d$ and $n < g_q(k,d)$ does not exists.
Suppose this is not true and let us find a contradiction. 
If $1 \le k<1+\log_q d$ then $\log_q d \le \log_q 2q = 1+\log_q 2$, and so $k$ can only be 1 or 2. 
We have already seen that if $k=1$ then $n < g_q(k,d)$ for any $n$, so suppose $k=2$. 
Suppose an $(n,q^2,d)_q$-systematic code exists with $n<\sum_{i=0}^1\left\lceil\frac{d}{q^i}\right\rceil=d+2$. 
Since by the Singleton bound $n \ge d+k-1$, then we can only have $n=d+1$, and therefore the only possible systematic code for which $n < g_q(2,d)$ must have parameter $(d+1,q^2,d)$, and so it is an MDS code. Let us call $C$ such a code. 
Being systematic, $C$ is the image of a map $F:(\FF_q)^2 \rightarrow (\FF_q)^{d+1}$ such that 
$F(x_1,x_2)= \left( x_1,x_2,f_3(x_1,x_2),\ldots,f_{d+1}(x_1,x_2) \right)$.
We can assume $F(0,0)=(0,\ldots,0)$. 
Any two codewords which have distance 1 in the two systematic components must have distance at least $d-1$ in the $d-1$ non-systematic components. 
Suppose there exists $\alpha,\beta_1,\beta_2 \in \FF_q,\beta_1 \ne \beta_2$ such that for a certain $i$ we have 
$f_i(\alpha,\beta_1)=f_i(\alpha,\beta_2) $.
In this case the distance between $F(\alpha,\beta_1)$ and $F(\alpha,\beta_2)$ is less than $d$. The same is true if we fix $\beta$ and we consider $\alpha_1$ and $\alpha_2$. This means that, whenever we fix $x_1=\alpha$ (respectively $x_2=\beta$) we need each $f_i(\alpha,x_2)$ (respectively $f_i(x_1,\beta)$) to be a permutation on $\FF_q$.
Due to this, for each fixed value $x_1=\alpha$, there exists a unique value $\beta$ such that $f_i(\alpha,\beta)=0$, for all $i$. 
Suppose now there exists $i\neq j$ such that $f_i(\alpha,\beta)=f_j(\alpha,\beta)=0$.
In this case the weight of $F(\alpha,\beta)$ is less than $d$, hence we have a contradiction (we assumed $0\in C$ and $\dd(C)=d$). 
We have obtained that if $f_i(\alpha,\beta)=0$, then $f_j(\alpha,\beta)\neq 0$ for all $j\neq i$.
We recall we have $f_3,\ldots, f_{d+1}$, and we have already proved that, for each fixed $\alpha$, there exists $\beta_1$ such that $f_3(\alpha,\beta_1)=0$. Hence if $f_4(\alpha,\beta_1)$ cannot be $0$ itself, there must exists another possible value $\beta_2$ such that $f_4(\alpha,\beta_2)=0$. Going on in this way we get a contradiction, in fact the number of $f_i$ is equal to $d-1$, and for them to be $0$ for different non-zero values $\beta_1,\ldots,\beta_{d-1}$, we need the field $\FF_q$ to contain at least $d$ different elements.
Hence we obtain the contradiction $q<d\leq 2q$ (by hypotesis) and $q\ge d$.
\end{proof}
\subsection{The case $q^{k-1} \mid d$}
 \label{section: q^(d-1)|d}
In this section we make use of the Plotkin bound to prove that there exists particular values of $d$ for which we can apply the Griesmer bound to nonlinear codes.
\begin{theorem}[Plotkin bound]\label{thm: plotkin}
Consider an $(n, M, d)_q$ code, with $M$ being the number of codewords in the code. If $n< \frac{qd}{q-1}$, then 
$M \leq d/(d-(1-1/q)n)$, or equivalently $n \ge d ((1-1/M)/(1-1/q))$.
\end{theorem}
\begin{proposition}\label{prop: griesmer 1}
For $r\ge 1$ it holds $N_q(q^k,q^{k-1}r)\ge g_q(k,q^{k-1}r)$.
\end{proposition}
\begin{proof}
Suppose there exists an $(n,q^k,q^{k-1}r)_q$-code $C$ that does not satisfies Griesmer bound. 
Hence $n<\sum_{i=0}^{k-1}\left\lceil\frac{q^{k-1}r}{q^i}\right\rceil$. 
Observe that in this case 
$
\sum_{i=0}^{k-1}\left\lceil\frac{q^{k-1}r}{q^i}\right\rceil=\sum_{i=0}^{k-1}\frac{q^{k-1}r}{q^i}
=
q^{k-1}r\sum_{i=0}^{k-1}\frac{1}{q^i}
$.
Since 
$\sum_{i=0}^{k-1}\frac{1}{q^i}=\frac{1-\frac{1}{q^k}}{1-\frac{1}{q}}$,
we obtain
\begin{equation}\label{eq:_no_griesmer}
n<q^{k-1}r\left(\frac{1-1/q^k}{1-1/q}\right)\,.
\end{equation}
We also observe that
$
n < 
q^{k-1}r
\left((1-1/q^k)/(1-1/q)\right)
<
q^{k-1}r
\left(1/(1-1/q)\right) = d/(1-1/q),
$
and we can write this inequality as $n<\frac{dq}{q-1}$,
which is the hypothesis for the Plotkin bound. 
Applying it, we get
$
q^k \leq 
\left\lfloor\frac{d}{d-n\left(1-1/q\right)}\right\rfloor
\leq
\frac{d}{d-n\left(1-1/q\right)}
$, 
i.e. $n\ge d\left(\frac{1-1/q^k}{1-1/q}\right)$, 
which contradicts equation \eqref{eq:_no_griesmer}. Hence each $(n,q^k,q^{k-1}r)_q$-code satisfies the Griesmer bound. \\
\end{proof}
Note that Proposition  \ref{prop: griesmer 1} is not restricted to systematic codes, but it holds for nonlinear codes with at least $q^k$ codewords, as next corollary explaines.
\begin{corollary}
Let $M\ge q^k$. For $r\ge 1$ it holds $N_q(M,q^{k-1}r)\ge g_q(k,q^{k-1}r)$.
\end{corollary}

\subsection{The case $d=rq^l, 1 \le r < q$}
\begin{lemma}\label{lem: small r}
Let $q$ be fixed, $d=q^lr$ for a certain $r$ such that $1\leq r< q$ and $l\ge 0$, and let $k$ such that $q^{k-1}\leq d$. Then $N_q(k,d)\ge g_q(k,d)$.
\end{lemma}
\begin{proof}
Being $1\leq r<q$, the hypothesis $q^{k-1}\leq d$ is equivalent to $k-1\leq l$. We use Proposition \ref{prop: griesmer 2} and we set $h=\min(k-1,l)$, obtaining
$
n\ge \sum_{i=0}^{k-1}\left\lceil\frac{d}{q^i}\right\rceil.
$
\end{proof}
\begin{theorem}\label{thm: griesmer qlr}
Let $1\leq r<q$ and $l$ a positive integer. Then $S_q(k,q^lr)\ge g_q(k,q^lr)$.
\end{theorem}
\begin{proof}
To prove that the Griesmer bound is true for these particular choices of $d$ we use Theorem \ref{thm: minimum k q}, 
hence we only need to prove that the Griesmer bound is true for all choices of $k$ such that $q^{k-1}\leq d$.
\\
We use now Lemma \ref{lem: small r}, which ensures that all such codes respect the Griesmer bound.
\end{proof}
\begin{corollary}\label{cor: griesmer 2l}
If $q=2$ then for each positive integer $l$ it holds $S_2(k,2^l)\ge g_2(k,2^l)$.
\end{corollary}
\begin{proof}
It follows directly from Theorem \ref{thm: griesmer qlr}.
\end{proof}
\subsection{The case $d=2^r - 2^s$}
%
%
In this section we prove that the Griesmer bound holds for all binary systematic codes whose distance is the difference of two powers of $2$. We need the following lemmas.
\begin{lemma}\label{lem: g(r+1)=2g(r)}
Let $r$ be a positive integer, and let $k\leq r+1$. Then $g_2(k,2^{r+1}) = 2 g_2(k,2^r)$.
\end{lemma}
\begin{proof}
The hypothesis $k\leq r+1$ implies that for any $i\leq k-1$, both $\left\lceil\frac{2^{r+1}}{2^i}\right\rceil=\frac{2^{r+1}}{2^i}$ and $\left\lceil\frac{2^{r}}{2^i}\right\rceil=\frac{2^{r}}{2^i}$. To prove our claim it is therefore enough to explicit $g_2(k,2^{r+1})$. Indeed we have
$$
g_2(k,2^{r+1})=\sum_{i=0}^{k-1}\left\lceil
\frac{2^{r+1}}{2^i}
\right\rceil
=\sum_{i=0}^{k-1}
\frac{2^{r+1}}{2^i}
=2\sum_{i=0}^{k-1}
\frac{2^{r}}{2^i}
=2\sum_{i=0}^{k-1}
\left\lceil\frac{2^{r}}{2^i}\right\rceil
=
2g_2(k,2^r)
$$
\end{proof}
\begin{lemma}\label{lem: increase g(k,d)}
For each $k$ and $d$ it holds
\begin{equation}\label{eq: inrease g(k,d)}
g_2(k,d+1)=g_2(k,d)+\min(k,l+1),
\end{equation}
where $l$ is the maximum integer such that $2^l$ divides $d$.
\end{lemma}
\begin{proof}
We consider $l$ as in the statement of the lemma, then $d=2^lr$, where $r$ is odd. 
We consider first the case $k\leq l+1$. The Griesmer bound for this choice of $k$ and $d$ is
\begin{equation}\nonumber
g_2(k,d+1)=\sum_{i=0}^{k-1}\left\lceil\frac{2^lr+1}{2^i}\right\rceil,
\end{equation}
and we observe that for each $i$ we have
$$
\left\lceil\frac{2^lr+1}{2^i}\right\rceil=
\frac{2^lr}{2^i}+ \left\lceil\frac{1}{2^i}\right\rceil=
\frac{2^lr}{2^i}+1=
\left\lceil\frac{2^lr}{2^i}\right\rceil+1.
$$
Therefore 
\begin{equation}\label{eq: k<l+1}
g_2(k,d+1)
=\sum_{i=0}^{k-1}\left(\left\lceil\frac{2^lr}{2^i}\right\rceil+1\right)
=g_2(k,d)+k.
\end{equation}
On the other hand, if $k>l+1$ we can split the sum in the following way:
\begin{equation}\label{eq: split sum}
g_2(k,d+1)=\left(\sum_{i=0}^{l}\left\lceil\frac{2^lr+1}{2^i}\right\rceil\right)
+
\left(\sum_{i=l+1}^{k-1}\left\lceil\frac{2^lr+1}{2^i}\right\rceil\right).
\end{equation}
For the first sum we make use of the same argument as above, while for the second sum we observe that $i>l$, which implies
$$
\left\lceil\frac{2^lr+1}{2^i}\right\rceil=\left\lceil\frac{2^lr}{2^i}\right\rceil.
$$
Putting together the two arguments, equation \eqref{eq: split sum} becomes
\begin{equation}\nonumber
g_2(k,d+1)=\left(\sum_{i=0}^{l}\left\lceil\frac{2^lr}{2^i}\right\rceil+l+1\right)
+
\left(\sum_{i=l+1}^{k-1}\left\lceil\frac{2^lr}{2^i}\right\rceil\right)
=
\sum_{i=0}^{k-1}\left\lceil\frac{2^lr}{2^i}\right\rceil+l+1,
\end{equation}
and the term on the right-hand side is $g_2(k,d)+l+1$. Together with \eqref{eq: k<l+1} this concludes the proof.
\end{proof}
\begin{lemma}\label{lem: g < 2^r}
If $k\ge r$, then $g_2(k,2^r)<2^{r+1}$.
\end{lemma}
\begin{proof}
Due to $k\leq r$, for $i<k$ it holds $\left\lceil\frac{2^{r}}{2^i}\right\rceil=\frac{2^{r}}{2^i}$. We can write the Griesmer bound as
$$
g_2(k,2^r)=\sum_{i=0}^{k-1}\frac{2^r}{2^i}=2^r\sum_{i=0}^{k-1}\frac{1}{2^i}<2^r\cdot 2.
$$
\end{proof}

\begin{theorem}\label{thm:griesmer_2r2s}
Let $r$ and $s$ be two positive integers such that $r>s$, and let $d=2^r-2^s$. Then $S_2(k,d)\ge g_2(k,d)$.
\end{theorem}
\begin{proof}
If $r=s+1$, then $2^r-2^s=2^s$, hence we can apply Corollary \ref{cor: griesmer 2} and our claim holds. Therefore we can assume $r\ge s+1$ in the rest of the proof.
Let us suppose there exists $s<r$ such that $S_2(k,2^r-2^s)< g_2(k, 2^r-2^s)$, 
i.e. the Griesmer bound does not hold for some $(n,2^k,d)_2$-systematic code $C$, with $d=2^r-2^s$ and $n = S_2(k,d)$. Due to Theorem \ref{thm: minimum k q}, we can consider the case $k<1+\log_2d$, so we put ourselves in the case $k\leq r$.
\\
We call $m$ the ratio $n/d$, which in the case of $C$ is
\begin{equation}\nonumber
m=\frac{S_2(k,2^r-2^s)}{2^r-2^s}
\leq \frac{g_2(k,2^r-2^s)-1}{2^r-2^s}
\end{equation}
We claim that $m<g_2(k,2^r)/(2^r)$. First we observe that if $k\leq r$, then 
\begin{equation}\nonumber
\frac{g_2(k,2^r)}{2^r}=
\sum_{i=0}^{k-1}\frac{1}{2^i}=
2\left(1-\frac{1}{2^k}
\right).
\end{equation}
We consider now the ratio $m$:
\begin{equation}\label{eq: m}
m \leq \frac{g_2(k,2^r-2^s)-1}{2^r-2^s}=
\frac{1}{2^r-2^s}\sum_{i=0}^{k-1}\left\lceil
\frac{2^{r}-2^s}{2^i}
\right\rceil
-\frac{1}{2^r-2^s}
\end{equation}
We start from the case $k\leq (s+1)$, and we can write \eqref{eq: m} as
\begin{equation}\nonumber
m < 
\frac{1}{2^r-2^s}\sum_{i=0}^{k-1}
\frac{2^{r}-2^s}{2^i}=
\sum_{i=0}^{k-1}\frac{1}{2^i}=
2\left(1-\frac{1}{2^k}
\right),
\end{equation}
so in this case $m<g_2(k,2^r)/(2^r)$. We consider now the case $k>s+1$, and we write our claim in the following equivalent way: 
$$
2^r(g_2(k,2^r-2^s)-1)<(2^r-2^s)g_2(k,2^r).
$$
Rearranging the terms we obtain
\begin{equation}\label{eq: equivalent claim m}
2^sg_2(k,2^r)<2^r(g_2(k,2^r)-g_2(k,2^r-2^s)+1),
\end{equation}
and we focus on the difference $g_2(k,2^r)-g_2(k,2^r-2^s)$. For any $d'$ in the range $2^r-2^s\leq d'<2^r$ we can apply Lemma \ref{lem: increase g(k,d)}, observing that $d'=2^lr$ where $l\leq s$, and this implies $k>l+1$. We obtain
$$
g_2(k,d'+1)=g_2(k,d')+l+1.
$$
Applying it for all distances from $2^r-2^s$ till we reach $2^r$ we obtain
\begin{equation}\label{eq: difference g}
g_2(k,2^r)-g_2(k,2^r-2^s)=2^{s+1}-1
\end{equation}
We substitute now \eqref{eq: difference g} into \eqref{eq: equivalent claim m}, which becomes
\begin{equation}\nonumber
2^sg_2(k,2^r)<2^r\cdot 2^{s+1}\qquad\Rightarrow \qquad g_2(k,2^r)<2^{r+1},
\end{equation}
and this is always true provided $k\leq r$, as shown in Lemma \ref{lem: g < 2^r}. 
\\
We now consider the $(tn,2^k,td)_2$-systematic code $C_t$ obtained by repeating $t$ times the code $C$. We remark that the value $m$ can be thought as the slope of the line $\dd(C_t)\mapsto \len(C_t)$, and we proved that $m<g_2(k,2^r)/(2^r)$. On the other hand, since $k\leq r$ we can apply Lemma \ref{lem: g(r+1)=2g(r)}, which ensures that $g_2(k,2^{r+b})=2^bg_2(k,2^r)$, namely the Griesmer bound computed on the powers of $2$ is itself a line, and its slope is strictly greater than $m$. Due to this we can find a pair $(t,\;b)$ such that
\begin{enumerate}
 \item  $td > 2^b$,
 \item  $tn < g_2(k, 2^b)$. 
\end{enumerate}
This means that we can find a systematic code with distance $d > 2^b$ and length $n<g_2(k,2^b)$. We can apply Lemma \ref{lem: reducing d}, and find a systematic code with the same length and distance equal to $2^b$, which means we have an $(tn,k,2^b)_2$-systematic code for which $n < g_2(k,2^b)$. 
This however contradicts Corollary \ref{cor: griesmer 2l}, hence for each $k\leq r$ we have 
$$
S_2(k,2^r-2^s)\ge g_2(k,2^r-2^s).
$$
Finally, observe that $k\leq r$ implies $k\leq \log_2(2^r)=\left\lceil\log_2(2^r-2^s)\right\rceil<1+\log_2d$, so we can apply Theorem \ref{thm: minimum k q} and conclude.
\end{proof}
\begin{corollary}\label{cor: odd distances}
Let $r$ and $s$ two positive integers such that $r>s$, and let $d=2^r-1$ or $d=2^r-2^s-1$. Then $S_2(k,d)\ge g_2(k,d)$.
\end{corollary}
\begin{proof}
We prove it for the case $d=2^r-2^s-1$, the same argument can be applied to the other case  by applying Corollary \ref{cor: griesmer 2l} instead of Theorem \ref{thm:griesmer_2r2s}.\\
Suppose $S_2(k,d)< g_2(k,d)$, i..e. there exists an $(n,k,d)_2$-systematic code for which 
\begin{equation}\label{eq: proof d odd}
n<g_2(k,d).
\end{equation}
We can extend such code to an $(n+1,k,d+1)_2$-systematic code $C$ by adding a parity check component to each codeword. $C$ has distance $\dd(C)=d+1 = 2^r-2^s$, so we can apply Theorem \ref{thm:griesmer_2r2s} to it, finding
$$
n+1\ge g_2(k,d+1)
$$
Observe that $d$ is odd, so applying Lemma \ref{lem: increase g(k,d)} we obtain
$$n+1\ge g_2(k,d+1) =g_2(k,d)+1\quad \Rightarrow \quad n\ge g_2(k,d),$$
which contradicts \eqref{eq: proof d odd}.
\end{proof}
%
%
%
%
\section{Versions of the Griesmer bound holding for nonlinear codes}
\label{sec:weak_griesmer}
In this section we provide some versions of the Griesmer bound holding for any systematic code. 
\subsection{An improvement of the Singleton bound}
For systematic codes we can improve the Singleton bound as follows.
\begin{proposition}\label{prop: singleton improved}
For any $k$ and $d$ it holds
$$
S_2(k,d)\ge k+\left\lceil\frac{3}{2}d\right\rceil-2.
$$
\end{proposition}
\begin{proof}
We will apply the same argument as for the proof of the Griesmer bound, which can be found in \cite{CGC-cd-book-huffmanPless03} (Section 2.4). \\
We consider a binary $(n=S_2(k,d),2^k,d)_2$-systematic code $C$ such that $0\in C$, and a codeword $c\in C$ whose weight is equal to the minimum distance $d$ of the code. 
We also assume $c$ has weight $1$ on its systematic part. 
The assumptions on $C$ and $c$ are w.l.o.g..\\ 
We construct a code $C'$ by puncturing $C$ in all the nonzero coordinates of $c$. We observe that $C'$ is itself a systematic code, due to the assumptions on $c$. In particular $C'$ is an $(n-d,2^{k-1},d')_2$-systematic code. We consider now a codeword $u\neq 0$ belonging to $C'$. There exists a vector $v\in\left(\mathbb{F}_2\right)^d$ such that the concatenation $(v|u)\in C$. This means that
$$
\left\{
\begin{array}{l}
w(v|u)=w(v)+w(u)\ge d\\
d(c,v|u)=d-w(v)+w(u)\ge d
\end{array}
\right.
$$
where $w(u)$ stands for the Hamming weight of $u$.
From the two inequality it follows that 
\begin{equation}\label{eq: w >d/2}
w(u)\ge \frac{d}{2} \;.
\end{equation}
We observe that \eqref{eq: w >d/2} is true for all non-zero codewords in $C'$, so we can choose $u$ to have weight $1$ in its systematic part. Therefore the length of $u$ has to be at least 
\begin{equation}\label{eq: length C'}
\len(u)\ge\frac{d}{2}+k-2\; .
\end{equation} 
Since $C'$ is an $(n-d,2^{k-1},d')_2$-systematic code, from \eqref{eq: length C'} we have
$$
n-d\ge\frac{d}{2}+k-2\qquad\Rightarrow\qquad n\ge k+\left\lceil\frac{3}{2}d\right\rceil -2
$$
\end{proof}

\subsection{Consequences of Theorem \ref{thm: griesmer qlr}}
We derive from Theorem \ref{thm: griesmer qlr} a weaker version of the Griesmer bound holding for any systematic code.
\begin{remark}\label{rem: exist l r}
Considering an integer $d$, there exist $1\leq r < q$ and $l\ge 0$ such that
\begin{equation}\label{eq: r l}
q^lr\leq d<q^l(r+1)\leq q^{l+1}
\end{equation}
In particular, $l$ has to be equal to $\left\lfloor \log_qd\right\rfloor$, and from inequality \eqref{eq: r l} we obtain
$d/q^l-1< r\leq d/q^l$,
namely $r=\left\lfloor d/q^l \right\rfloor.$
\end{remark}
\begin{corollary}[Bound A]\label{cor: griesmer qlr systematic}
Let $l=\left\lfloor\log_qd\right\rfloor$ and $r=\left\lfloor d/q^l \right\rfloor$. Then
$$
S_q(k,d)\ge d+\sum_{i=1}^{k-1}\left\lceil \frac{q^lr}{q^i} \right\rceil.
$$
\end{corollary}
\begin{proof}
We call $s$ the difference between $d$ and $q^lr$, namely $d=q^lr+s$. Note that $s\leq n-k$, and so there are at least $s$ non-systematic coordinates.
With this notation, let $C$ be an $(n,q^k,q^lr+s)_q$-systematic code. 
We build a new code $C_s$ by puncturing $C$ in $s$ systematic coordinates, 
so that $C_s$ has parameters $(n-s, q^k, d_s)_q$, for a certain $q^lr\leq d_s\leq q^lr+s$.\\
If $q^lr\neq d_s$, we can apply Lemma \ref{lem: reducing d}, in order to obtain another code $\bar{C}$, 
so that we have an $(n-s,q^k,q^lr)_q$-systematic code. 
Due to Remark \ref{rem: exist l r} it holds $1\leq r < q$, so we can apply Theorem \ref{thm: griesmer qlr} to $\bar{C}$. 
We find
$n-s\ge \sum_{i=0}^{k-1}\left\lceil q^lr/q^i \right\rceil$,
hence 
$n\ge \sum_{i=0}^{k-1}\left\lceil q^lr/q^i \right\rceil+s$.
We conclude by noticing that for $i=0$ we have $\left\lceil \frac{q^lr}{q^i}\right\rceil=q^lr$, and by adding $s$ we obtain exactly $d$. 
So
$ n\ge d+\sum_{i=1}^{k-1}\left\lceil q^lr/q^i \right\rceil$.
\end{proof}
\subsection{Consequences of Proposition \ref{prop: griesmer 1} }
Next we generalize Proposition \ref{prop: griesmer 1}. 
\begin{proposition}\label{prop: griesmer 2}
Let $q$, $k$ and $d$ be fixed, and let $l$ be the maximum integer such that $q^l$ divides $d$. 
Then it holds 
$$
N_q(q^k,d) \ge \sum_{i=0}^{h}\left\lceil \frac{d}{q^i}\right\rceil,
$$
where $h$ is the minimum between $k-1$ and $l$. 
\end{proposition}
\begin{proof}
First, notice that $d=q^lr,\quad q\nmid r$.
We can use the same argument as for the proof of Proposition \ref{prop: griesmer 1}. 
If $k-1 | l$, then we are in the same situation as above. 
Otherwise $h=l$, and $d$ is not divisible for higher powers of $q$, and we need to stop the sum to the term $\frac{d}{q^l}$.
\end{proof}
\begin{corollary}[Bound B]\label{cor: griesmer 2}
Let $q$, $M$ and $d$ be fixed, let $k$ the maximum integer such that $q^{k} \leq M$, and let $l$ be the maximum integer such that $q^l$ divides $d$.
Then it holds
$$
N_q(M,d) \ge \sum_{i=0}^{h}\left\lceil \frac{d}{q^i}\right\rceil,
$$
where $h$ is the minimum between $k-1$ and $l$. 
\end{corollary}
\begin{proof}
If there exists an $(n,M,d)_q$-code, then there exists also an $(n,q^k,d)_q$ code, due to the condition $q^k\leq M$. 
Hence we can apply Proposition \ref{prop: griesmer 2}.
\end{proof}
%
\subsection{Relations between the Griesmer bound and the Plotkin bound}
We consider now the following bounds, which can be seen as weaker versions of the Griesmer bound or as an extension of the Plotkin bound.
\begin{proposition}\label{prop: weak griesmer}
For each choice of $q$, $k$ and $d$, it holds
$$
N_q(q^k,d) \ge \left\lceil \sum_{i=0}^{k-1} \frac{d}{q^i} \right\rceil.
$$
\end{proposition}
\begin{proof}
We can use an argument similar to the proof of Proposition \ref{prop: griesmer 1}.
Suppose there is a code $C$ such that 
$n < \left\lceil \sum_{i=0}^{k-1} \frac{d}{q^i}\right\rceil$.
Observe that 
$
\sum_{i=0}^{k-1} d/q^i 
=
d
\left(
\frac{1-1/q^k}{1-1/q}
\right)
< d
\left(
\frac{1}{1-1/q}
\right)
$,
i.e.
$
n< d
\left(
\frac{1}{1-1/q}
\right)
$,
which allow us to apply the Plotkin bound and to find the contradiction
$
n \ge d\left(
\frac{1-1/q^k}{1-1/q}
\right)
$.
\end{proof}
From a direct computation, as we did in the proof, we find that Proposition \ref{prop: weak griesmer} can be also written as
\begin{proposition}\label{prop: explicit weak Griesmer}
For each choice of $q$, $k$ and $d$, it holds
\begin{equation}\label{eq: explicit weak griesmer}
N_q(q^k,d)\ge \left\lceil 
d\left(
\frac{1-\frac{1}{q^k}}{1-\frac{1}{q}}
\right)
\right\rceil.
\end{equation}
\end{proposition}
Observe that if the code has a number of words $M \ge q^k$, then by removing $M - q^k$ codewords 
we obtain an $(n,q^k,d)_q$-code and we can apply Proposition \ref{prop: explicit weak Griesmer}. 
We obtain the following Corollary.
\begin{corollary}[Bound C]
For each choice of $q$, $k$ and $d$, it holds
\begin{equation}\label{eq: explicit weak griesmer, K words}
N_q(N,d)\ge \left\lceil 
d\left(
\frac{1-\frac{1}{q^k}}{1-\frac{1}{q}}
\right)
\right\rceil.
\end{equation}
where $k$ is the larger integer such that $M \ge q^k$.
\end{corollary}
%
%
\section{Counterexamples to the Griesmer bound}
\label{sec:counterexample}
In this section we show explicitly binary nonlinear codes for which the Griesmer bound does not hold. It is indeed already known that there exist pairs $(k,d)$ for which $N_2(2^k,d)<g_2(k,d)$, however it was not clear whether the same was true for systematic codes or not. We start in the next section by expliciting a nonlinear non-systematic code whose length contradicts the Griesmer bound. Then we make use of this code to explicit a systematic code contradicting itself the Griesmer bound, proving that in general $g_2(k,d)$ is not a bound for systematic codes.
\subsection{The nonlinear case}
In \cite{CGC-art-cod-Levenshtein64application}, Levenshtein has shown that if Hadamard matrices of certain orders exist, 
then the binary codes obtained from them meet the Plotkin Bound. 
Levenshtein's method to construct such codes can be found also in the proof of Theorem 8, of \cite[Ch.~3,\S 2]{CGC-cd-book-macwilliamsTOT}.\\
\\
\begin{example}\label{ex:nonlinear_counterex}
 The next code is a $(19,16,10)_2$-nonlinear and non-systematic code, obtained using Levensthein's method, as explained in \cite[Ch.~3,\S 2]{CGC-cd-book-macwilliamsTOT}. 
 All its codewords have weight 10 (except the zero codeword) and each pair of codewords has distance $d=10$. 
 The code is composed by the zero codeword and by 15 shifts of the codeword $c=(1,1,0,0,1,1,1,1,0,1,0,1,0,0,0,0,1,1,0)$. We show here explicitly the code:
 {\footnotesize 
 \begin{align*}
  C = \{
& (0,0,0,0,0,0,0,0,0,0,0,0,0,0,0,0,0,0,0) , \\
& (1,1,0,0,1,1,1,1,0,1,0,1,0,0,0,0,1,1,0) , \\
& (1,0,0,1,1,1,1,0,1,0,1,0,0,0,0,1,1,0,1) , \\
& (0,0,1,1,1,1,0,1,0,1,0,0,0,0,1,1,0,1,1) , \\
& (0,1,1,1,1,0,1,0,1,0,0,0,0,1,1,0,1,1,0) , \\
& (1,1,1,1,0,1,0,1,0,0,0,0,1,1,0,1,1,0,0) , \\
& (1,1,1,0,1,0,1,0,0,0,0,1,1,0,1,1,0,0,1) , \\
& (1,1,0,1,0,1,0,0,0,0,1,1,0,1,1,0,0,1,1) , \\
& (1,0,1,0,1,0,0,0,0,1,1,0,1,1,0,0,1,1,1) , \\
& (0,1,0,1,0,0,0,0,1,1,0,1,1,0,0,1,1,1,1) , \\
& (1,0,1,0,0,0,0,1,1,0,1,1,0,0,1,1,1,1,0) , \\
& (0,1,0,0,0,0,1,1,0,1,1,0,0,1,1,1,1,0,1) , \\
& (1,0,0,0,0,1,1,0,1,1,0,0,1,1,1,1,0,1,0) , \\
& (0,0,0,0,1,1,0,1,1,0,0,1,1,1,1,0,1,0,1) , \\
& (0,0,0,1,1,0,1,1,0,0,1,1,1,1,0,1,0,1,0) , \\
& (0,0,1,1,0,1,1,0,0,1,1,1,1,0,1,0,1,0,0) 
\}
 \end{align*} 
 }
This code has length $n = 19 < g_2(4,10) = 20$, i.e. the code $C$ proves that $N_2(16,10)<g_2(4,10)$.
\end{example}
\subsection{The systematic case}
In this section we provide an example of an $(n,q^k,d)_2$-systematic code for which $n < g_2(k,d)$, proving that in general the Griesmer bound does not hold for systematic codes. 
\begin{example}
 To construct an $(n,k,d)_2$-systematic code for which $n < g_q(k,d)$, we search for a $[15,4,8]_2$-linear code $C_l$. We remark that $C_l$ would attain the Griesmer bound with equality, and being $d=8$, we can apply Corollary \ref{cor: griesmer 2l} to be sure that no binary nonlinear systematic codes exists with the same dimension and distance but smaller length.\\ To build $C_l$ we consider the cyclic code of length $15$ associated to the complete defining set $S=\{0,1,2,3,4,5,6,8,9,10,12\}$, which is a code with dimension $4$ and distance $8$. We can therefore find a systematic linear code equivalent to $C_l$. A possible choice is the code generated by 
 \begin{equation}\nonumber
 \setcounter{MaxMatrixCols}{20}
\begin{bmatrix}
1 & 0 & 0 & 0 & 1 & 1 & 1 & 0 & 1 & 0 & 0 & 1 & 0 & 1 & 1  \\
0 & 1 & 0 & 0 & 1 & 1 & 0 & 1 & 1 & 0 & 1 & 0 & 1 & 0 & 1  \\
0 & 0 & 1 & 0 & 1 & 0 & 1 & 1 & 1 & 1 & 0 & 0 & 1 & 1 & 0  \\
0 & 0 & 0 & 1& 0  & 1 & 1 & 1 & 1 & 1 & 1 & 1 & 0 & 0 & 0  \\
\end{bmatrix}
\end{equation}
To obtain a nonlinear systematic code $\bar{C}$ not verifying the Griesmer bound we make use of both this code and the code $C$ in Example \ref{ex:nonlinear_counterex}. This new code is obtained by concatenating each codeword in $C_l$ with a different codeword in $C$. In this way $\bar{C}$ is an $(34,4,18)_2$-systematic code. 
In the following we explicit all codewords in $\bar{C}$.
 {\footnotesize 
 \begin{align*}
  \bar{C} = \{
& (0,0,0,0,0,0,0,0,0,0,0,0,0,0,0,\;\; 0,0,0,0,0,0,0,0,0,0,0,0,0,0,0,0,0,0,0) , \\
& (1,0,0,0,1,1,1,0,1,0,0,1,0,1,1,\;\; 1,1,0,0,1,1,1,1,0,1,0,1,0,0,0,0,1,1,0) , \\
& (1,1,0,0,0,0,1,1,0,0,1,1,1,1,0,\;\; 1,0,0,1,1,1,1,0,1,0,1,0,0,0,0,1,1,0,1) , \\
& (0,1,0,0,1,1,0,1,1,0,1,0,1,0,1,\;\; 0,0,1,1,1,1,0,1,0,1,0,0,0,0,1,1,0,1,1) , \\
& (0,1,1,0,0,1,1,0,0,1,1,0,0,1,1,\;\; 0,1,1,1,1,0,1,0,1,0,0,0,0,1,1,0,1,1,0) , \\
& (1,1,1,0,1,0,0,0,1,1,1,1,0,0,0,\;\; 1,1,1,1,0,1,0,1,0,0,0,0,1,1,0,1,1,0,0) , \\
& (1,0,1,0,0,1,0,1,0,1,0,1,1,0,1,\;\; 1,1,1,0,1,0,1,0,0,0,0,1,1,0,1,1,0,0,1) , \\
& (0,0,1,0,1,0,1,1,1,1,0,0,1,1,0,\;\; 1,1,0,1,0,1,0,0,0,0,1,1,0,1,1,0,0,1,1) , \\
& (0,0,1,1,1,1,0,0,0,0,1,1,1,1,0,\;\; 1,0,1,0,1,0,0,0,0,1,1,0,1,1,0,0,1,1,1) , \\
& (1,0,1,1,0,0,1,0,1,0,1,0,1,0,1,\;\; 0,1,0,1,0,0,0,0,1,1,0,1,1,0,0,1,1,1,1) , \\
& (1,1,1,1,1,1,1,1,0,0,0,0,0,0,0,\;\; 1,0,1,0,0,0,0,1,1,0,1,1,0,0,1,1,1,1,0) , \\
& (0,1,1,1,0,0,0,1,1,0,0,1,0,1,1,\;\; 0,1,0,0,0,0,1,1,0,1,1,0,0,1,1,1,1,0,1) , \\
& (0,1,0,1,1,0,1,0,0,1,0,1,1,0,1,\;\; 1,0,0,0,0,1,1,0,1,1,0,0,1,1,1,1,0,1,0) , \\
& (1,1,0,1,0,1,0,0,1,1,0,0,1,1,0,\;\; 0,0,0,0,1,1,0,1,1,0,0,1,1,1,1,0,1,0,1) , \\
& (1,0,0,1,1,0,0,1,0,1,1,0,0,1,1,\;\; 0,0,0,1,1,0,1,1,0,0,1,1,1,1,0,1,0,1,0) , \\
& (0,0,0,1,0,1,1,1,1,1,1,1,0,0,0,\;\; 0,0,1,1,0,1,1,0,0,1,1,1,1,0,1,0,1,0,0) 
\}
 \end{align*} 
 }
Notice that $g_2(4,18) = 35$, therefore $S_2(4,18)<g_2(4,18)$, proving that the Griesmer bound is in general not true for systematic codes.\\
\end{example}
We conjecture the following:
\begin{conjecture}
For any $r\ge 3$ there is a systematic code with distance $2^r+2$ not satisfying the Griesmer bound.
\end{conjecture}
The example given is a special case with $r=4$.

%
%
\section{Conclusions}\label{sec:conclusion}
In this work we have addressed the problem of characterize the Griesmer bound for systematic nonlinear codes, mainly in the binary case. The Griesmer bound is one of the few bounds which can only be applied to linear codes, however classical counterexamples arose from the Levensthein's method for building optimal nonlinear codes, which however does not provide specific counterexamples for the systematic case. It was therefore non fully understood whether the Griesmer bound would hold for systematic nonlinear codes, or whether there exist families of parameters $(k,d)$ for which the bound could be applied to the nonlinear case. Moreover, weaker versions of the Griesmer bound might hold for nonlinear codes.
\\
As regards nonlinear codes satisfying the Griesmer bound, the main results of our work are Theorem \ref{thm:griesmer_2r2s} and Corollary \ref{cor: odd distances}, in which we prove that whenever a binary systematic nonlinear code has a distance $d$ such that
\begin{enumerate}
\item $d=2^r$,
\item $d=2^r-1$,
\item $d=2^r-2^s$, or
\item $d=2^r-2^s-1$,
\end{enumerate}
then the Griesmer bound can be applied.\\
We also provide versions of the Griesmer bound holding for nonlinear codes: Bound A, Bound B and Bound C.\\
Finally, we conclude by showing explicit examples of systematic nonlinear codes for which the Griesmer bound does not hold.\\
All the results can be easily extended to codes over any alphabet.

\section{Aknowledgements}\label{sec:conclusion}
The first three authors would like to thank their (former) supervisor: the last author.

%
%

\bibliographystyle{amsalpha}
\bibliography{RefsCGC}

\end{document}